\documentclass[a4paper,11pt,reqno]{amsart}

\usepackage{amsmath}
\usepackage{enumerate}
\usepackage{amsfonts}
\usepackage{gensymb}
\usepackage{wasysym}
\usepackage{amsthm}
\usepackage{enumerate}
\usepackage{amssymb}
\usepackage{tikz}
\usepackage{stmaryrd}
\usepackage{slashbox}
\usepackage{fullpage}

\newcommand{\U}{\mathcal{U}}

\newcommand{\upto}{\upharpoonright}  
\newcommand{\MLR}{\text{MLR}}
\newcommand{\dom}{\mathrm{dom}}
\newcommand{\halts}{\downarrow}
\newcommand{\univ}{U}

\newcommand{\cs}{2^{\omega}}
\newcommand{\str}{2^{<\omega}}

\newcommand{\llb}{\llbracket}
\newcommand{\rrb}{\rrbracket}

\newtheorem{thm}{Theorem}[section]
\newtheorem{lem}[thm]{Lemma}

\newtheorem{cor}[thm]{Corollary}
\newtheorem{prop}[thm]{Proposition}

\newtheorem{que}[thm]{Question}

\theoremstyle{definition}
\newtheorem{dfn}[thm]{Definition}

\theoremstyle{remark}
\newtheorem{rmk}[thm]{Remark}

\usepackage{xcolor}	
	\usepackage{soul}

\title{Kolmogorov Complexity and Generalized Length Functions}
\author{Cameron Fraize and Christopher P.\ Porter}

\begin{document}

\maketitle

\begin{abstract}  
Kolmogorov complexity measures the algorithmic complexity of a finite binary string $\sigma$ in terms of the length of the shortest input $\sigma^*$ to a fixed universal machine that yields $\sigma$ as output. In this article, we study a modification of Kolmogorov complexity by replacing the notion of the length of input with what we refer to as a \emph{generalized length function}.  We focus in particular on length functions defined in terms of a computable real number $\alpha\geq 1$ such that for a fixed $i\in\{0,1\}$, the length of $i$ is 1 and the length of $1-i$ is $\alpha$ (and extends this to arbitrary finite strings); such a length function is referred to as an $\alpha$-\emph{length} function.  As we will see, the class of $\alpha$-length functions for computable $\alpha\geq 1$ is closely related to the family of computable Bernoulli measures.  We study randomness with respect to Bernoulli measures through the lens of this family of length functions, proving a generalization version of the classic Levin-Schnorr theorem that involves $\alpha$-length functions and then proving subsequent results that involve effective dimension and entropy, including a generalization of the effective Shannon-McMillan-Breiman theorem for sequences that are random with respect to a computable Bernoulli measure.

\end{abstract}

\section{Introduction}

According to the Levin-Schnorr theorem, a sequence $X$ is Martin-L\"of random if and only if the initial segments of $X$ are sufficiently incompressible; that is,  there is some $c$ such that for all $n$,
$K(X\upto n)\geq n-c$, where $K(\sigma)$ is the prefix-free Kolmogorov complexity of a binary string $\sigma$.  This expression says that the length of the shortest input into a fixed universal prefix-free Turing machine that yields $X\upto n$  is at least $n-c$.

The motivating question for this study is this:  Does the Levin-Schnorr theorem still hold if we modify the notion of length in terms of which $K$ is defined?  Here we examine a specific family of computable length functions defined as follows:  For computable $\alpha\in\mathbb{R}$ with $\alpha\geq 1$, $i\in\{0,1\}$, and $\sigma\in\str$, we set $\ell^i_\alpha(\sigma)=\#_{1-i}(\sigma)+\alpha\cdot\#_i(\sigma)$; informally, we can think of such a length function as being determined by the cost of individual bits, so that the bit $i\in\{0,1\}$ has cost $\alpha$ while the bit $1-i$ has cost $1$.  

The main finding of this paper is that an analogue of the Levin-Schnorr theorem holds for all such length functions. Moreover, this analogue provides a characterization of Bernoulli randomness with respect to some computable Bernoulli parameter $p$ (which we define in Section \ref{sec-background} below) in terms of what we call $\alpha$-complexity, written $K^{(\alpha)}$, an analogue of Kolmogorov complexity that uses one of the length functions $\ell^i_\alpha$ for $i\in\{0,1\}$.  In the course of establishing our main result, we also provide a generalization of the KC theorem (also referred to as the machine existence theorem in \cite{Nie09}), obtained by a proof that is simpler than the standard proofs of the KC theorem in \cite{Nie09} and \cite{DowHir10}.  Lastly, by means of our generalization of the Levin-Schnorr theorem and related auxiliary results, we generalize an effective version of the Shannon-McMillan-Breiman theorem, which involves the effective Hausdorff dimension of Bernoulli random sequences.  The upshot of these results is that, in the context of algorithmic randomness, there is an interesting and fruitful connection between computable probability measures on the one hand and various computable length functions on the other, one that merits further study in future work.

The contents of this paper are as follows.  In Section \ref{sec-background}, we introduce the requisite background in computability theory and algorithmic randomness.  In Section \ref{sec-glf}, we introduce the concept of a generalized length function and define a variant of Kolmogorov complexity in terms of generalized length functions, focussing specifically on the $\alpha$-functions described above and the corresponding notion of $\alpha$-complexity.



Next, in Section \ref{sec-MLR}, we study possible generalizations of the classic Levin-Schnorr theorem in terms of generalized length functions.  For computable $\alpha,\beta\geq 1$, we arrive at a characterization of a class of sequences that are Martin-L{\"o}f random with a respect to a Bernoulli measure with parameter $p$ satisfying $p^\beta=1-p$ in terms of $\alpha$-complexity of their initial segments, modulo a multiplicative constant that accounts for the differences between $\alpha$-length and $\beta$-length.  We also show that the above-mentioned multiplicative constant in our generalization of the Levin-Schnorr theorem is necessary by showing that there are sequences $X$ such that (i) the initial segments of $X$ have $\alpha$-complexity above a threshold that does not include the multiplicative constant but (ii) are not random with respect to any computable measure.

Lastly, in Section \ref{sec-dim}, we modify the notions of effective packing dimension and effective Hausdorff dimension using $\alpha$-complexity and proof the above-mentioned generalization of the effective Shannon-McMillan-Breiman theorem for sequences that are random with respect to a computable Bernoulli measure.

\bigskip

\section{Background}\label{sec-background}
 
\subsection{Notation}We fix the following notation and terminology.  We denote the set of natural numbers by $\omega$ and the set of finite binary strings by $2^{<\omega}$, with $\epsilon$ denoting the empty string.  We shall use lowercase letters such as $n,m$ to denote natural numbers, and lowercase Greek letters such as $\sigma$ and $\tau$ to denote binary strings.  All logarithms without subscripts (i.e. $\log x$ for $x>0$) will be base 2, unless otherwise stated.  We use $2^\omega$ to denote \textit{Cantor space}, the set of infinite binary sequences, and use capital letters such as $X$ and $Y$ to denote such sequences.  The set of non-negative dyadic rationals, of the form $m/2^n$ for $m,n\in\omega$ is denoted $\mathbb{Q}_2$.  If $X\in 2^\omega$ and $n\in\omega$, then $X\upto n$ is the first $n$ bits of $X$, and $X(n)$ is the $(n+1)^{\text{st}}$ bit of $X$.  If $\sigma$ and $\tau$ are binary strings, then $\sigma\preceq\tau$ means that $\sigma$ is an initial segment of $\tau$, i.e. $\tau\upto|\sigma|=\sigma$.  Similarly, for $X\in 2^\omega$, $\sigma\prec X$ means that $\sigma$ is an initial segment of $X$, and the \textit{cylinder set} $\llbracket\sigma\rrbracket$ is the set of all $X\in 2^\omega$ such that $\sigma\prec X$.  For strings $\sigma$ and $\tau$, $\sigma^\smallfrown\tau$ is the concatenation of $\sigma$ and $\tau$.  Given $\sigma$, $\#_0(\sigma)$ is the number of 0's in $\sigma$, and $\#_1(\sigma)$ is the number of $1$'s in $\sigma$.  If $S\subseteq 2^{<\omega}$, then we let $\llbracket S\rrbracket$ denote $\cup_{\sigma\in S}\llbracket\sigma\rrbracket$.  The cylinder sets form a basis for the usual topology on Cantor space (the product topology), and so the open sets in $2^\omega$ are of the form $\llbracket S\rrbracket$ for $S\subseteq 2^{<\omega}$.  An open set $\U$ is said to be \textit{effectively open} (or $\Sigma^0_1$) if there is a computably enumerable (hereafter, c.e.) set $S\subseteq 2^{<\omega}$ such that $\U=\llbracket S\rrbracket$.  A sequence $\{\U_n\}_{n\in\omega}$ is said to be \textit{uniformly} $\Sigma^0_1$ if there exists a sequence $\{S_n\}_{n\in\omega}$ of uniformly c.e. sets such that $\U_n=\llbracket S_n\rrbracket$.  

We assume that the reader is familiar with the basics of computability theory and algorithmic randomness, i.e., computable functions, prefix-free machines, universal machines, etc.  See, for instance,
 Nies~\cite[Ch.~1]{Nie09} or Downey and Hirschfeldt~\cite[Ch.~2]{DowHir10}.

\subsection{Computable measures on $\cs$} Recall that a measure $\mu$ on $\cs$ is computable if there is a computable function $f:2^{<\omega}\times\omega\to\mathbb{Q}_2$ such that
$|\mu(\llbracket\sigma\rrbracket)-f(\sigma,i)|\leq 2^{-i}$.  For a prefix-free $V\subseteq\str$ (i.e., for $\sigma\in V$, if $\sigma\prec\tau$, then $\tau\notin V$), we set $\mu(\llb V\rrb)=\sum_{\sigma\in V}\mu(\sigma)$.  Hereafter, we denote $\mu(\llbracket\sigma\rrbracket)$ by $\mu(\sigma)$ for strings $\sigma$, and $\mu(\llbracket V\rrbracket)$ by $\mu(V)$ for $V\subseteq 2^{<\omega}$.  We also denote the \textit{Lebesgue measure} by $\lambda$, where $\lambda(\sigma)=2^{-|\sigma|}$ for every string $\sigma\in\str$.

Given $p\in(0,1)$, the \textit{Bernoulli} $p$\textit{-measure on} $2^\omega$ is defined by $\mu_p(\sigma)=p^{\#_1(\sigma)}(1-p)^{\#_{0}(\sigma)}$ for each $\sigma\in\str$. (Note that $\lambda$ is the Bernoulli $(1/2)$-measure.)  A Bernoulli $p$-measure is computable if and only if $p$ is a computable real number.

\subsection{Martin-L\"of randomness and Kolmogorov complexity} Recall that for a computable measure $\mu$ on $2^\omega$, a \textit{$\mu$-Martin-L{\"o}f test} is a uniformly $\Sigma^0_1$ sequence $\{\U_i\}_{i\in\omega}$ such that $\mu(\U_i)\leq 2^{-i}$ for every $i\in\omega$.  A sequence $X\in 2^\omega\!$ \textit{passes} a $\mu$-Martin-L{\"o}f test $\{\U_i\}_{i\in\omega}$ if $X\not\in\cap_{i\in\omega} \U_i$, and $X$ is called \textit{$\mu$-Martin-L{\"o}f random}, written $X\in\MLR_\mu$, if $X$ passes every $\mu$-Martin-L{\"o}f test.

We can characterize $\mu$-Martin-L\"of randomness for any computable measure $\mu$ in terms of Kolmogorov complexity.  Recall that the {\em prefix-free Kolmogorov complexity} of a string $\sigma \in \str$ is defined as the length of the shortest program producing $\sigma$ by a fixed universal, prefix-free machine $U$, that is, as $K(\sigma)=\min\{|\tau|\colon \tau \in \str \;\&\; U(\tau)=\sigma\}.$  The relationship  between Martin-L\"of randomness and Kolmogorov complexity is given by the following:
\begin{thm}[Levin \cite{Lev74}, Schnorr; see \cite{Cha75}]
A sequence $X\in 2^\omega\!$ is $\mu$-Martin-L{\"o}f random iff there is some $c\in\omega$ such that for all $n\in\omega$,
\[
K(X\upto n)\geq -\log(\mu(X\upto n))-c.
\]
\end{thm}
Note that in the case that $\mu=\lambda$, we have that $X$ is Martin-L\"of random if and only if there is some $c\in\omega$ such that $K(X\upto n)\geq n-c$ for all $n\in\omega$.

\subsection{Notions of effective dimension}  The following dual concepts measure the density of information in binary sequences.  For $X\in 2^\omega\!$, the \textit{effective Hausdorff dimension of} $X$ is 
$\liminf_{n\to\infty} \frac{K(X\upto n)}{n}$, while the \textit{effective packing dimension of} $X$ is
$\limsup_{n\to\infty} \frac{K(X\upto n)}{n}$.

It follows from the Levin-Schnorr theorem that for any $X\in\MLR$,
\[
\liminf_{n\to\infty} \frac{K(X\upto n)}{n}=\limsup_{n\to\infty} \frac{K(X\upto n)}{n}=\lim_{n\to\infty} \frac{K(X\upto n)}{n}=1.
\]
Hoyrup \cite{Hoy12} generalized this result to a wide class of computable measures, namely the shift-invariant computable measures (which we will not define here). In particular, it follows from Hoyrup's result that for every computable Bernoulli measure $\mu_p$ (i.e., every Bernoulli measure defined in terms of a computable parameter $p\in(0,1)$) and every $X\in \MLR_{\mu_p}$,
\[
\lim_{n\rightarrow\infty}\frac{K(X\upto n)}{n}=h(p),
\]
where $h(p)=-p\log(p)-(1-p)\log(1-p)$ is the entropy of the measure $\mu_p$.

\section{Generalized length functions and Kolmogorov complexity}\label{sec-glf}

In this section, we introduce the concept of a generalized length function and a corresponding modification of Kolmogorov complexity.

\begin{dfn}
Let $\ell:2^{<\omega}\to\mathbb{R}^{\geq 0}$ be a function.

\begin{enumerate}[(i)]
\item $\ell$ is called a \textit{generalized length function}, or \textit{g.l.f.}, if $\ell$ is computable, $\ell(\epsilon)=0$, and for $\sigma,\tau\in\str\!$, if $\sigma\prec\tau$, then $\ell(\sigma)<\ell(\tau)$;
\item $\ell$ is a \textit{sub-additive g.l.f.} if $\ell(\sigma\tau)\leq\ell(\sigma)+\ell(\tau)$ for every $\sigma,\tau\in\str$;
\item $\ell$ is \textit{additive} if equality holds for in the above condition.
\end{enumerate}
\end{dfn}

For a g.l.f.\ $\ell$, a prefix-free machine $M$, and a string $\sigma\in\mathrm{ran}(M)$, the \textit{prefix-free $\ell$-Kolmogorov complexity of} $\sigma$ \textit{relative to} $M$ is defined to be
\[
K^{(\ell)}_M(\sigma)=\{\ell(\tau):M(\tau)\halts=\sigma\},
\]
where $K_M^{(\ell)}(\sigma)=\infty$ if $\sigma\notin\mathrm{ran}(M)$.  Given a prefix-free universal machine $\univ$, the \textit{$\ell$-Kolmogorov complexity of $\sigma$} is $K^{(\ell)}=K^{(\ell)}_\univ$.

If a g.l.f.\ $\ell$ is sub-additive, we can prove that the invariance theorem holds for $K^{(\ell)}$.

\begin{prop}[Invariance]
For any sub-additive g.l.f.\ $\ell$ and any prefix-free machine $M$, for all $\sigma$, $K^{(\ell)}(\sigma)\leq K_M^{(\ell)}(\sigma)+O(1)$.
\end{prop}

\noindent As this proof follows in the same way as the proof of the invariance theorem for $K$, we leave the details to the reader.

%
%
%
%

\subsection{$\alpha$-length functions}\label{subsec-bern}

One specific family of g.l.f.'s that we will study here consists of what we call the \emph{$\alpha$-length functions}. For $\alpha\in\mathbb{R}$ with $\alpha\geq 1$ and $i\in\{0,1\}$, let $\ell^i_\alpha$ be the function defined by
\[
\ell^i_\alpha(\sigma)=\#_{1-i}(\sigma)+\alpha\cdot\#_i(\sigma)
\]
for $\sigma\in\str$. Hereafter, we set the convention that $\ell_\alpha=\ell_\alpha^0$ and will refer to $\ell_\alpha(\sigma)$ as the \emph{$\alpha$-length} of $\sigma$.  (All results involving $\ell_\alpha^0$ established here also hold for $\ell_\alpha^1$ taking the appropriate symmetries into account.)  We will also write $K^{(\ell_\alpha)}(\sigma)$ as $K^{(\alpha)}(\sigma)$ and will refer to this as the \emph{$\alpha$-complexity} of $\sigma$.

The reason we study is particular class of g.l.f.'s is that the class of Bernoulli $p$-measures is closely tied to the family of $\alpha$-length functions.  Let $p\in[\frac{1}{2},1)$, and set $\alpha=\log(1-p)/\log(p)$, so that $p^\alpha=1-p$.  Then for $\sigma\in\str$, we have
\[
\lambda_p(\sigma)=p^{\#_1(\sigma)}(1-p)^{\#_0(\sigma)}=p^{\#_1(\sigma)}p^{\alpha\cdot\#_0(\sigma)}=p^{\ell_\alpha(\sigma)}.
\]
Thus we have a generalization of the equation $\lambda(\sigma)=(\frac{1}{2})^{|\sigma|}$ for all Bernoulli measures.  Hereafter, we will focus on computable $p\in(0,1)$ (and hence computable values $\alpha=(\log(p)/\log(1-p))$).  Note also that if $p< 1/2$, then we simply work with $\ell^1_\alpha$ instead of $\ell^0_\alpha$, so that all of the results obtained below for the case that $p\geq 1/2$ still hold.

In the case that $\alpha\in\omega$, we use the variable $k$ instead of $\alpha$ and refer to the corresponding length functions as $k$-length functions.  Note that if $p^k=1-p$, then $p$ is the unique solution in $(0,1)$ to the equation $f_k(x)=x^k+x-1$.  In particular, if $k=2$, then the corresponding solution $p\in(0,1)$ to $f_2(x)=x^2+x-1$ is $\phi-1$, where $\phi=\frac{1+\sqrt{5}}{2}$ is the golden mean.

\subsection{Some inequalities involving $\alpha$-length}

For $\alpha\geq 1$ we are interested in determining the number of strings of a certain $\alpha$-length for a fixed $\alpha\geq 1$.  For every $\alpha\geq 1$, let $S_\alpha^n=\{\sigma\in\str\colon\ell_\alpha(\sigma)\in[n,n+1)\}$, and let $s^n_\alpha=|S^n_\alpha|$.  In general, it is difficult to find an exact expression to calculate the values $(s^n_\alpha)_{n\in\omega}$, with an exception in the case of $k$-length:  For $k=2$, the sequence $(s^n_2)_{n\in\omega}$ is the sequence of Fibonacci numbers, and for integers $k>2$, the sequence $(s^n_k)_{n\in\omega}$ is the sequence of so-called $k$-Fibonacci numbers (where $s^n_k=1$ for $n<k$ and $s^n_k=s^{n-1}_k + s^{n-k}_k$ for $n\geq k$).

\begin{prop}\label{prop-bound1}
Let $p\in(0,1)$ and let $\alpha$ satisfy $p^\alpha=1-p$.  Then for every $n\in\omega$,
\[
\mu_p(S^n_\alpha)\geq 1-p.
\]
\end{prop}

\begin{proof}
Suppose for the sake of contradiction that $\mu_p(S^n_\alpha) < 1- p$ for some n.  Since $S^n_\alpha$ is finite, there is a finite set of strings $D\subseteq\str$ such that  $\cs\setminus\llb S^n_\alpha\rrb =\llb D\rrb$ and hence $\mu_p(\llb D\rrb)>p$.  We can further assume that the collection $D$ is prefix-free and that every proper initial segment of each $\tau\in D$ extends to some $X\in\llb S^n_\alpha\rrb$.

We make several observations.  First, note that no initial segment of any $\tau\in D$ is in $S^n_\alpha$, as the set of all extensions of strings in $D$ is disjoint from the extensions of strings in $S^n_\alpha$.  Second, for each $\tau\in D$, $\ell_\alpha(\tau)\geq n+1$.  Indeed, suppose that $\ell_\alpha(\tau)< n+1$.  We cannot have $\ell_\alpha(\tau)\geq n$, or else we would have $\tau\in S^n_\alpha$.  However, if $\ell_\alpha(\tau)<n$, then there is some $k\in\omega$ such that $\tau^\frown 0^k\in S^n_\alpha$, which contradicts our assumption that $\llb D\rrb$ and $\llb S^n_\alpha\rrb$ are disjoint.

Finally, the key observation is that each $\tau\in D$ ends with a 0.  For if there is some $\tau\in D$ that ends with a 1, then dropping this final 1 gives a string $\tau^-$ such that $\ell_\alpha(\tau^-)\geq n$ (since $\ell_\alpha(\tau)\geq n+1$) which extends to a string in $S^n_\alpha$ by our initial assumption about the strings in $D$.  But this implies that $\tau^-\in S^n_\alpha$, which contradicts the fact that no initial segment of $\tau$ can be in $S^n_\alpha$.

Now,  consider the set $D'=\{\tau^-\colon \tau\in D\}$, which is prefix-free since $D$ is prefix-free.  Then
\[
\mu_p\left(\bigcup_{\tau\in D}\llb\tau\rrb\right)=\sum_{\tau\in D}\mu_p(\tau)=\sum_{\tau\in D}(1-p)\mu_p(\tau^-)=(1-p)\sum_{\tau^-\in D'}\mu_p(\tau^-)\leq 1-p.
\]
But $\mu_p\left(\bigcup_{\tau\in D}\llb\tau\rrb\right)> p$, which, combined with the above inequality implies that $1-p > p$, a contradiction.

\end{proof}

The other inequality involving the size of the sets $S^n_\alpha$ is the following:

\begin{prop}\label{prop-bound2}
For all $\alpha\geq 1$ and $n\in\omega$, $s_\alpha^n\leq p_\alpha^{-n}$.
\end{prop}

\begin{proof}
In the case that $\alpha=1$, $p_\alpha=\frac{1}{2}$ and the result clearly follows, so we assume that $\alpha>1$.  Note that by our assumption that $p_\alpha> 1-p_\alpha$, it follows that $1\leq p_\alpha^{-1}\leq 2$.  Now for $n=0$ we have $s_\alpha^0=1=p_\alpha^0$.  Moreover, since $\alpha> 1$, $S^1_\alpha=\{1\}$, and so $s^1_\alpha\leq p^{-1}_\alpha$.   In this case, it follows that $\mu_{p_\alpha}(S^1_\alpha)=p_\alpha$, from which it follows that $\mu_p(S^n_\alpha)<p_\alpha$ for all $n\geq 2$ (as the sets $S^n_\alpha$ are disjoint for $n\in\omega$).

 Assume for the sake of contradiction that there is some $n>1$ such that $s^n_\alpha>p_\alpha^{-n}$.  Then since for every $\sigma\in S^n_\alpha$ we have $\mu_{p_\alpha}(\sigma)=p_\alpha^{\ell_\alpha(\sigma)}\geq p_\alpha^{n+1}$, we have
 \[
 \mu_{p_\alpha}(S^n_\alpha)=\sum_{\sigma\in S^n_\alpha}\mu_{p_\alpha}(\sigma)\geq s^n_\alpha\cdot p_\alpha^{n+1}\geq p_\alpha^{-n}p_\alpha^{n+1}=p_\alpha,
\]
a contradiction.  Thus the conclusion follows.
\end{proof}

\subsection{Generalizing the KC Thereom}

The bounds established in Propositions \ref{prop-bound1} and \ref{prop-bound2}  above are particularly useful in generalizing other results involving $K$ to $K^{(\alpha)}$, notably the KC theorem and the minimality of $K$ as an information content measure.  Let us consider each of these results in turn.


\begin{thm}[$\alpha$-KC Theorem]
Let $\alpha\geq 1$ be computable and let $\{(r_i,\tau_i)\}_{i\in\omega}$ be a computable sequence of pairs (called requests) with $r_i\in\omega$ and $\tau_i\in 2^{<\omega}$ for every $i$, such that $\sum_{i\in\omega} p_\alpha^{r_i}\leq 1-p_\alpha$. Then there exists a prefix-free machine $M$ and sequence $\{\sigma_i\}_{i\in\omega}$ of strings with $\ell_\alpha(\sigma_i)\in [r_i,r_i+1)$ such that $\dom(M)=\{\sigma_i:i\in\omega\}$ and $M(\sigma_i)=\tau_i$ for every $i$.  Furthermore, one can obtain an index for $M$ effectively from an index of the sequence of requests.
\end{thm}

\bigskip

Note that the bounded sum condition imposed on the requests differs from the bound in the original KC theorem: instead of requiring $\sum_{i\in\omega} p_\alpha^{r_i}\leq 1$, we have $\sum_{i\in\omega} p_\alpha^{r_i}\leq 1-p_\alpha$.  To see why, first note that, following the measure-theoretic proof given in  \cite{Nie09}, we can always construct effectively a prefix-free machine $M$ and sequence of strings $\{\sigma_i\}_{i\in\omega}$ under the condition that $\sum_{i\in\omega} 2^{-|\sigma_i|}=\sum_{i\in\omega} 2^{-r_i}\leq 1$.  


This approach no longer works in the general case of $\alpha$-length. For example, let $\alpha=2$ and consider the finite set of requests $\{(2,\tau_0),(4,\tau_1),(4,\tau_2),(4,\tau_3),(4,\tau_4)\}$ 
for some strings $\tau_0,\dotsc,\tau_4\in\str$. We can compute that 
\[
\sum_{i=0}^4 p_2^{r_i}=\sum_{i=0}^4 \phi^{-r_i}\leq 1,
\]
where, recall, $\phi$ is the golden mean.  However, there is no prefix-free set $\{\sigma_0,\ldots,\sigma_4\}$ fulfilling these requests.  This is because, as can be easily verified, once we will fulfill the request of a string of 2-length 2, there are only three strings of 2-length 4 available for additional requests.

We resolve the problem of requesting too much measure by ensuring that our request sets never cover more than any set of the form $S^n_k$ for any $n\in\omega$.  Indeed, for every $\alpha\geq1$ and $n\in\omega$, according to Proposition \ref{prop-bound1} we have $\mu_{p_\alpha}(S^n_\alpha)\geq 1-p_\alpha$, which is equal to the bound on requests in the statement of the $\alpha$-KC theorem.

%

Lastly, note that in the case that $\alpha=1$, the above statement of the $\alpha$-KC theorem differs from the original KC theorem by imposing the bound 1/2 on the requests rather than 1.  As we will see, using this bound allows us to prove the 
result with a proof that is considerably simpler than that of the original KC theorem.  Indeed, in the standard proofs of the KC theorem, one has to carefully choose each string to fulfill each request (choosing the leftmost available string of the requested length); however, if we use the bound 1/2, then we are guaranteed that we can choose \emph{any} available string that satisfies a given request (which only comes at a cost of increasing the length of each request in the original KC set by one bit).  This follows from the following general lemma.

\begin{lem}\label{KClem}
For $\alpha\geq 1$, suppose that $\{\sigma_0,\ldots,\sigma_n\}$ is a prefix-free set such that $r=\sum_{i=0}^n p_\alpha^{\ell_\alpha(\sigma_i)}\leq 1-p_\alpha$.  Then for every $j\in\omega$, if
$p_\alpha^j+r\leq 1-p_\alpha$, then there is some $\tau\in 2^{<\omega}$ such that $\ell_\alpha(\tau)\in[j,j+1)$ and $\{\sigma_1,\ldots,\sigma_n,\tau\}$ is prefix-free.
\end{lem}

\begin{proof}
Suppose not, so that there is some $j\in\omega$ such that $p_\alpha^j+r\leq 1-p_\alpha$ but for all strings $\tau$, if $\ell_\alpha(\tau)\in[j,j+1)$, then there is some $0\leq t\leq n$ such that $\sigma_t\preceq\tau$.  But then since $\llbracket S_\alpha^j\rrbracket\subseteq\bigcup_{i=0}^n\llbracket\sigma_i\rrbracket$, we have
\[
r=\sum_{i=0}^n p_\alpha^{\ell_\alpha(\sigma_i)}=\mu_{p_\alpha}\left(\bigcup_{i=0}^n\llbracket\sigma_i\rrbracket\right)\geq\mu_{p_\alpha}(\llbracket S_\alpha^j\rrbracket)\geq 1-p_\alpha,
\]
where the final inequality is given by Proposition \ref{prop-bound1}.  But then ${p_\alpha}^j+r\geq {p_\alpha}^j+1-p_\alpha>1-p_\alpha$, a contradiction.
\end{proof}

\begin{proof}[Proof of the $\alpha$-KC theorem]
Given a computable sequence $\{(r_i,\tau_i)\}_{i\in\omega}$ of requests such that $\sum_{i\in\omega} q_k^{-r_i}\leq q_k^{-k}$, we define a sequence $(\sigma_i)_{i\in\omega}$ and a prefix-free machine $M$ by recursion.  Let $\sigma_0=1^{r_0}$, so that $\ell_\alpha(\sigma_0)=r_0$.  Given $n\in\omega$, if $\sigma_0,\ldots,\sigma_n$ are all defined, search for a string $\tau$ such that $\ell_\alpha(\tau)\in[r_{n+1},r_{n+1}+1)$ and $\sigma_i\not\preceq\tau$ for any $1\leq i\leq n$.  This process terminates since such a string $\tau$ exists by Lemma \ref{KClem}, and the process is computable since $\ell_\alpha$ is computable.  Set $\sigma_{n+1}=\tau$ and define $M(\sigma_{n+1})=\tau_{n+1}$. This concludes the construction.
\end{proof}

Just as with standard complexity, we can define a notion of information content measure for $K^{(\alpha)}$ and generalize the fact that for any information content measure $F$, there is a $c$ such that for any $\sigma\in\dom(F)$, $K(\sigma)\leq F(\sigma)+c$.  This will be especially useful in giving sufficient conditions for $\mu_{p_\alpha}$-Martin-L{\"o}f randomness in Section \ref{sec-MLR}.

\begin{dfn}
An \emph{$\alpha$-information content measure} (hereafter an \emph{$\alpha$-i.c.m.}) is a partial map $F:2^{<\omega}\to\mathbb{R}^{\geq 0}$ such that $\sum_{\sigma\in\dom(F)} p_\alpha^{F(\sigma)}\leq 1$ and the set $\{(\sigma,m):F(\sigma)< m\}$ is c.e.
\end{dfn}

\begin{rmk}
Unlike the definition of a standard i.c.m., in defining an $\alpha$-i.c.e.\ $F$, we use a strict inequality in the condition that the set $\{(\sigma,m):F(\sigma)< m\}$ is c.e. (rather than $\leq$), since in this more general setting, $F$ is not integer-valued and thus the condition $F(\sigma)=m$ need not be verifiable in a finite number of stages.
\end{rmk}

As with traditional i.c.m.'s, $\alpha$-complexity $K^{(\alpha)}$ is identifiable as the minimal $\alpha$-i.c.m., by noting that for every $\alpha$-i.c.m.\ $F$, the set $S=\{(m+\lceil\alpha\rceil,\sigma):F(\sigma)< m\}$ is an $\alpha$-KC set.  Indeed, the set $S$ is c.e., and, setting 
\[
S_1=\{m+\lceil\alpha\rceil\in\omega\colon \exists\,\sigma\in\str\; (m+\lceil\alpha\rceil,\sigma)\in S\},
\]
we have that
\begin{align*}
\sum_{m+\lceil\alpha\rceil\in S_1} p_\alpha^{m+\lceil\alpha\rceil} &\leq \sum_{\sigma\in\dom(F)} p_\alpha^{F(\sigma)+\alpha} &&\text{(since $p_\alpha <1$)}\\
&=p_\alpha^\alpha\cdot\sum_{\sigma\in\dom(F)} p_\alpha^{F(\sigma)}\\
&\leq p_\alpha^\alpha = 1-p_\alpha.\\
\end{align*}
We can thus build a prefix-free machine such that for every $\sigma\in\dom(F)$, there is a string $\tau$ such that $\ell_\alpha(\tau)\in \bigl[\lceil F(\sigma)\rceil+\lceil\alpha\rceil,\lceil F(\sigma)\rceil +\lceil\alpha\rceil+1\bigr)$ and $M(\tau)=\sigma$. Thus, for every $\alpha\geq 1$ and every $\alpha$-i.c.m.\ $F$, there exists some $c\in\omega$ such that for all $\sigma\in\dom(F)$, $K^{(\alpha)}(\sigma)\leq F(\sigma)+c$.

Using the fact that $K^{(\alpha)}$ is minimal among all $\alpha$-i.c.m.'s, we can also find an upper bound for the $\alpha$-complexity of all strings $\sigma$ in terms of $\beta$-length for some $\beta\geq 1$.  Here we introduce a multiplicative term $\log p_\beta/\log p_\alpha$ that we will make use of repeatedly in Sections \ref{sec-MLR} and \ref{sec-dim}, which functions as a conversion factor between $\alpha$-length and $\beta$-length.

\begin{thm}\label{thm-K-upper-bound}
For any $k\geq 1$, there is some $c\in\omega$ such that for all $\sigma$,

\[
K^{(\alpha)}(\sigma)\leq \frac{\log p_\beta}{\log p_\alpha}\left(\ell_\beta(\sigma)+K^{(\beta)}(\lfloor\ell_\beta(\sigma)\rfloor)\right)+c.
\]

\end{thm}

\begin{proof}
Let $\alpha\geq 1$.  It is enough to show that the map $F:\str\rightarrow\omega$ defined by $F(\sigma)=\left(\log p_\beta/\log p_\alpha\right)(\ell_\beta(\sigma)+K^{(\beta)}(\lfloor\ell_\beta(\sigma)\rfloor)$ is an $\alpha$-i.c.m., which is uniformly right-c.e.\ (i.e., computably approximable from above uniformly in $\sigma$) since $\alpha$ and $\beta$ are computable and $K^{\beta}$ is uniformly right-c.e.  Recall that for all $n\in\omega$, $|S^n_k|\leq p_\alpha^{-n}$ by Proposition \ref{prop-bound2}.  Then
since $p_\alpha^{\log p_\beta/\log p_\alpha}=p_\alpha^{\log_{p_\alpha} p_\beta}=p_\beta$, we have
\[
\begin{split}
\sum_{\sigma\in 2^{<\omega}} p_\alpha^{\frac{\log p_\beta}{\log p_\alpha}\left[\ell_\beta(\sigma)+K^{(\beta)}(\lfloor\ell_\beta(\sigma)\rfloor)\right]} &=
\sum_{\sigma\in 2^{<\omega}} p_\beta^{\ell_\beta(\sigma)+K^{(\beta)}(\lfloor\ell_\beta(\sigma)\rfloor)}\\
&\leq\sum_{n\in\omega} |S_\beta^n|\cdot p_\beta^{n+K^{(\beta)}(n)}\\
&\leq\sum_{n\in\omega} p_\beta^{-n} p_\beta^{n+K^{(\beta)}(n)}\\
&=\sum_{n\in\omega} p_\beta^{K^{(\beta)}(n)}\\
&\leq\sum_{\sigma\in 2^{<\omega}} p_\beta^{K^{(\beta)}(\sigma)}\\
&\leq\sum_{\tau\in\dom(\univ)} p_\beta^{\ell_\beta(\tau)}\\
&= \mu_{p_\beta}(\dom(\univ))\leq 1.\\
\end{split}
\]
Moreover, since $F$ is uniformly right-c.e., the set $\{(\sigma,m):F(\sigma)< m\}$ is c.e.  The result then follows by minimality of $K^{(\alpha)}$ among $\alpha$-i.c.m.'s.
\end{proof}

\begin{rmk}\label{rmk-jk}
Note that for $\beta\leq \alpha$, $p_\beta\leq p_\alpha$, so $\log p_\beta\leq\log p_\alpha$.  Hence the factor $\log p_\beta/\log p_\alpha\leq 1$ for $\beta\leq \alpha$, and, similarly, $\log p_\beta/\log p_\alpha>1$ for $\beta>\alpha$.
\end{rmk}

\section{$K^{(\alpha)}$-incompressibility and Randomness}\label{sec-MLR}

We now consider the notion of $K^{(\alpha)}$-incompressibility in the context of infinite sequences.

\begin{dfn}
For computable $\alpha,\beta\geq 1$, a sequence $X\in\cs$ is \emph{$(\alpha,\beta)$-incompressible} if
\[
K^{(\alpha)}(X\upto n)\geq \ell_\beta(X\upto n)-O(1).
\]
\end{dfn}

Just as we can characterize Martin-L\"of randomness in terms of incompressibility, it is natural to consider whether a similar result holds
for $(\alpha,\beta)$-incompressibility.  First, we have the following result.

\begin{thm}\label{thm-kj-none}
For $\alpha,\beta\geq 1$, if $\alpha<\beta$, then there are no $(\alpha,\beta)$-incompressible sequences.
\end{thm}

To show this, we will require the following lemma:

\begin{lem}\label{lem-K-inequality}
For every computable $\alpha\geq 1$, there is some $c\in\omega$ such that for all $n\geq 1$, $K^{(\alpha)}(n)\leq (\alpha+1)\log n + c$.
\end{lem}

\begin{proof}
Let $\alpha\geq 1$ be computable.  For each $m$, let $\{\sigma_i^m\}_{i\leq 2^m}$ be a strictly-increasing enumeration of $2^m$ with respect to the lexicographic ordering.  We use the prefix-free encoding of the positive integers given by, for every $n\geq 1$, $\rho_n=1^{\lfloor\log n\rfloor+1} 0 \sigma_i^{\lfloor\log n\rfloor+1}$\!, where $\sigma_i^{\lfloor \log n\rfloor+1}$ is the binary representation of $n$.   Given $n$, we have that
\[
\begin{split}
\ell_\alpha(\rho_n) &=\ell_\alpha(1^{\lfloor\log n\rfloor+1} 0 \sigma_i^{\lfloor\log n\rfloor+1})\\
&= \ell_\alpha(1^{\lfloor\log n\rfloor+1})+\ell_\alpha(0)+\ell_\alpha(\sigma_i^{\lfloor\log n\rfloor+1})\\
&=\lfloor\log n\rfloor+1+\alpha+\ell_\alpha(\sigma_i^{\lfloor\log n\rfloor+1})\\
&\leq \lfloor\log n\rfloor+1+\alpha+\ell_\alpha(0^{\lfloor\log n\rfloor+1})\\
&= \lfloor\log n\rfloor+1+\alpha+\alpha(\lfloor\log n\rfloor+1)\\
&\leq (\alpha+1)\lfloor\log n\rfloor + 2\alpha+1.
\end{split}
\]

\noindent Now, by invariance for $K^{(\alpha)}$, let $c\in\omega$ be such that for all $n\geq 1$, $K^{(\alpha)}(n)\leq \ell_\alpha(\rho_n)+c$.  Thus $K^{(\alpha)}(n)\leq (\alpha+1)\lfloor\log n\rfloor + 2\alpha+1+c$, and the result follows.
\end{proof}

\begin{proof}[Proof of Theorem \ref{thm-kj-none}]
For $\alpha<\beta$, suppose that $X\in\cs$ is $(\alpha,\beta)$-incompressible, so that there is some $d$ such that for all $n\in\omega$,
\begin{equation}\label{eq3}
K^{(\alpha)}(X\upto n)\geq \ell_\beta(X\upto n)-d.
\end{equation}
By Theorem \ref{thm-K-upper-bound}, let $c$ be such that for all $n\in\omega$,
\begin{equation}\label{eq4}
K^{(\alpha)}(X\upto n) \leq\frac{\log p_\beta}{\log p_\alpha}\left(\ell_\beta(X\upto n)+K^{(\beta)}(\lfloor\ell_\beta(X\upto n)\rfloor)\right)+c.
\end{equation}
Combining (\ref{eq3}) and (\ref{eq4}) and rearranging yields
\begin{equation}\label{eq5}
\Biggl(1-\frac{\log p_\beta}{\log p_\alpha}\Biggr)\ell_\beta(X\upto n)\leq \frac{\log p_\beta}{\log p_\alpha}K^{(\beta)}(\lfloor\ell_\beta(X\upto n)\rfloor)+c+d.
\end{equation}
Applying Lemma \ref{lem-K-inequality} to (\ref{eq5}) yields
\[
\Biggl(1-\frac{\log p_\beta}{\log p_\alpha}\Biggr)\ell_\beta(X\upto n)\leq \frac{\log p_\beta}{\log p_\alpha}(\alpha+1)\log(\lfloor\ell_\beta(X\upto n)\rfloor)+c+d
\]
which implies that
\begin{equation}\label{eq6}
\frac{\Bigl(1-\frac{\log p_\beta}{\log p_\alpha}\Bigr)}{\frac{\log p_\beta}{\log p_\alpha}(\alpha+1)}\ell_\beta(X\upto n)\leq \log(\ell_\beta(X\upto n))+c+d.
\end{equation}
By Remark \ref{rmk-jk}, since $\alpha<\beta$, we have $\frac{\log p_\beta}{\log p_\alpha}<1$, so setting $r=\frac{\log p_\beta}{\log p_\alpha}$, (\ref{eq6}) becomes
\begin{equation}\label{eq7}
\frac{1-r}{r(\alpha+1)}\ell_\beta(X\upto n)\leq \log(\ell_\beta(X\upto n))+c+d,
\end{equation}
where the term on the left-hand side of the inequality is positive. Let $N$ be the least such that for all $n\geq N$,
\[
\frac{1-r}{r(\alpha+1)}n> \log(n)+c+d;
\]
such an $N$ exists by a routine calculation.  Since the function $n\mapsto \ell_\beta(X\upto n)$ is unbounded, we can find some $n_0$ such that $\ell_\beta(X \upto n_0)\geq N$, so that
\[
\frac{1-r}{r(k+1)}\ell_j(X\upto n_0)> \log(\ell_\beta(X\upto n_0))+c+d,
\]
which contradicts (\ref{eq7}).  Thus, no $(\alpha,\beta)$-incompressible sequences exist.
\end{proof}

Next, we investigate $(\alpha,\beta)$-incompressible sequences for $\alpha\geq \beta$.  In particular, we show that there is a connection between a sequence being $(\alpha,\beta)$-incompressible and its being complex.  

\begin{dfn}
$X\in\cs$ is \emph{complex} if there is some computable order (that is, a computable, unbounded, non-decreasing function) $f:\omega\rightarrow\omega$ such that 
\[
K(X\upto n)\geq f(n).
\]
\end{dfn}

We now show that for $\alpha\geq \beta$, all $(\alpha,\beta)$-incompressible sequences are complex.  We will make use of an effective version of the law of large numbers established by Davie in \cite{Dav01}.  Let $\mathcal{K}_c=\{X\in\cs\colon(\forall n)K(X\upto n)\geq n-c\}$.  

\begin{thm}[Davie \cite{Dav01}]\label{thm-Davie}
For any $c\in\omega$ and $\epsilon>0$, we can effectively find $n(c, \epsilon)\in\omega$ such that if $X\in \mathcal{K}_c$ then for all $n>n(c,\epsilon)$,
\[
\Biggl|\frac{\#_0(X\upto n)}{n}-\frac{1}{2}\Biggr|<\epsilon.
\]
\end{thm}

Theorem \ref{thm-Davie} does not just apply to infinite sequences but also to sufficiently long incompressible strings.  Indeed, for a fixed $c\in\omega$ and $\epsilon>0$, Theorem \ref{thm-Davie} provides a bound $n(c,\epsilon)$ such that any $c$-incompressible string $\sigma$ of length exceeding $n(c,\epsilon)$ satisfies the above condition on $\#_0(\sigma)$.  We will apply Theorem \ref{thm-Davie} to a specific collection of sufficiently incompressible finite strings.

For each $\sigma$, let $\sigma^*$ be the lexicographically least string such that $U(\sigma^*)=\sigma$.

\begin{lem}\label{lem-davie}
For $\epsilon\in(0,1)$, there is some $c$ and $n(c,\epsilon)$ such that for all strings $\sigma$ such that $|\sigma^*|\geq n(c,\epsilon)$, 
\[
\Biggl|\frac{\#_0(\sigma^*)}{|\sigma^*|}-\frac{1}{2}\Biggr|<\epsilon.
\].
\end{lem}

\begin{proof}
Let $c$ be the coding constant of the machine $U\circ U$.  We claim that for sufficiently long strings $\sigma$, $\sigma^*$ is $c$-incompressible.  Indeed, if there is some $\sigma$ such that $|\sigma^*|>c$ and $|\sigma^*|$ is $c$-compressible, then there is some $\tau$ such that $U(\tau)=\sigma^*$ and $|\tau|<|\sigma^*|-c$.  But then $U(U(\tau))=\sigma$, so that
\[
K(\sigma)\leq K_{U\circ \,U}(\sigma)+c\leq |\tau| +c<|\sigma^*|=K(\sigma),
\]
a contradiction.  Applying Davie's theorem to any string $\sigma^*$ with $|\sigma^*|\geq n(c,\epsilon)$ yields the conclusion.
\end{proof}

\begin{thm}
For computable $\alpha\geq \beta\geq 1$, every $(\alpha,\beta)$-incompressible sequence is complex.
\end{thm}

\begin{proof}
For $\alpha\geq \beta$, let $X$ be $(\alpha,\beta)$-incompressible.  Let $d\in\omega$ be such that for every $n\in\omega$, if $U(\tau)=X\upto n$, then
$\ell_\alpha(\tau)\geq\ell_\beta(X\upto n)-d$.  By definition of $\ell_\alpha$ and $\ell_\beta$, this yields
\begin{equation}\label{eq-kj1}
|\tau|+(\alpha-1)\cdot \#_1(\tau)=\ell_\alpha(\tau)\geq \ell_\beta(X\upto n)-d\geq n-d.
\end{equation}

Fix some rational $\epsilon\in (0,1)$.  For all sufficiently large $n\in\omega$, we have $K(X\upto n)\geq n(c,\epsilon)$, where $c$ is the coding constant of $U\circ U$.  Applying Lemma \ref{lem-davie}  to $\tau=(X\upto n)^*$ yields
\[
\Biggl|\frac{\#_0(\tau)}{|\tau|}-\frac{1}{2}\Biggr|<\epsilon,
\]
which implies that 
\begin{equation}\label{eq-kj2}
\#_1(\tau)\leq (1/2+\epsilon)|\tau|.
\end{equation}
From (\ref{eq-kj2}) it follows that
\[
|\tau|+(\alpha-1)\cdot\#_1(\tau)\leq |\tau|+(\alpha-1)(1/2+\epsilon)|\tau|.
\]
Combining this inequality with (\ref{eq-kj1}) and the fact that $|\tau|=K(X\upto n)$ yields
\[
K(X\upto n)\bigl(1+(\alpha-1)(1/2+\epsilon)\bigr)\geq n-d,
\]
from which it follows that 
\[
K(X\upto n)\geq\frac{n-d}{1+(\alpha-1)(1/2+\epsilon)}
\]
for some $e\in\omega$. Setting 
\[
f(n)=\Biggl\lfloor\frac{n-d}{1+(\alpha-1)(1/2+\epsilon)}\Biggr\rfloor
\]
(and setting $f(n)=0$ if $n<d$) yields the desired computable order that witnesses that $X$ is complex.
\end{proof}

As we will show shortly, $(\alpha,\beta)$-incompressibility does not guarantee randomness, at least for the case that $\beta=1$:  there are $(\alpha,1)$-incompressible sequences that are not random with respect to any computable measure.  However, if we modify the definition of $(\alpha,\beta)$-incompressibility for any $\alpha$ and $\beta$ (even when $\alpha<\beta$), we do get a characterization of randomness.

\begin{dfn}

For computable $\alpha,\beta\geq 1$, a sequence $X\in\cs$ is \emph{generalized $(\alpha,\beta)$-incompressible} if
\[
K^{(\alpha)}(X\upto n)\geq \frac{\log p_\beta}{\log p_\alpha}\ell_\beta(X\upto n)-O(1).
\]
\end{dfn}

\begin{thm}\label{thm-gen-kj}
For any $X\in 2^\omega$ and computable $\alpha,\beta\geq 1$, $X\in\MLR_{\mu_{p_\beta}}$ if and only if $X$ is generalized $(\alpha,\beta)$-incompressible.
\end{thm}

To prove Theorem \ref{thm-gen-kj}, we need two lemmas.

\begin{lem} \label{lem-kj1}
For every $\alpha\geq 1$,

\[
\sum_{n\in\omega} p_\alpha^{\alpha n+1}= 1.
\]

\end{lem}

\begin{proof}
\[
\sum_{n\in\omega} p_\alpha^{\alpha n+1}=p_\alpha\sum_{n\in\omega}(1-p_\alpha)^n=\dfrac{p_\alpha}{1-(1-p_\alpha)}=1.
\]
\end{proof}

\begin{lem}\label{lem-kj2}
Let $M$ be a prefix-free machine, let $n\in\omega$, let $\alpha,\beta\geq 1$ be computable, and let 
\[
R_n=\left\{\sigma:K_M^{(\alpha)}(\sigma)<\frac{\log p_\beta}{\log p_\alpha}(\ell_\beta(\sigma)-n)\right\}.
\]
Then $\mu_{p_\beta}(R_n)\leq p_\beta^n\cdot\mu_{p_k}(\dom(M))$.
\end{lem}

\begin{proof}
For every $\sigma\in R_n$, let $\tau_\sigma\in\dom(M)$ be such that $M(\tau_\sigma)\halts=\sigma$ and
\[
\ell_\alpha(\tau_\sigma)<\frac{\log p_\beta}{\log p_\alpha}(\ell_\beta(\sigma)-n).
\]
Noting, as in the proof of Theorem \ref{thm-K-upper-bound}, that $p_\beta^{\log p_\alpha/\log p_\beta}=p_\beta^{\log_{p_\beta} p_\alpha}=p_\alpha$, we have
\[
\begin{split}
\mu_{p_\beta}(R_n) &= \sum_{\sigma\in R_n} p_\beta^{\ell_\beta(\sigma)}<\sum_{\sigma\in R_n} p_\beta^{\frac{\log p_\alpha}{\log p_\beta}\ell_\alpha(\tau_\sigma)+n}\\
&= p_\beta^{n}\sum_{\sigma\in R_n} p_\alpha^{\ell_\alpha(\tau_\sigma)}\leq p_\beta^{n}\sum_{\tau\in\dom(M)} p_\alpha^{\ell_\alpha(\tau)}\\
&=p_\beta^{n}\cdot\mu_{p_\alpha}(\dom(M)).
\end{split}
\]
\end{proof}

\begin{proof}[Proof of Theorem \ref{thm-gen-kj}]($\rightarrow$): Let $X\in 2^\omega$ and suppose that $X\in\MLR_{\mu_{p_\beta}}$. For any $n$, let
\[
R_n=\left\{\sigma\in\str:K^{(\alpha)}(\sigma)<\frac{\log p_\beta}{\log p_\alpha}(\ell_\beta(\sigma)-n)\right\}.
\]
By Lemma \ref{lem-kj2}, the sequence $\{\U_n\}_{n\in\omega}$ such that $\U_n=\llbracket R_n\rrbracket$ for every $n$ is uniformly $\Sigma^0_1$ and $\mu_{p_{\beta}}(\U_n)\leq p_\beta^n$ for every $n\in\omega$.  Since $\{p_\beta^{n}\}_{n\in\omega}$ is a computable, decreasing sequence of computable reals, it follows that a computable subsequence of $\{\U_n\}_{n\in\omega}$ is a $\mu_{p_\beta}$-Martin-L{\"o}f test, and so $X\not\in \U_c$ for some $c\in\omega$.  Thus, for all $n$,
\[
K^{(\alpha)}(X\upto n)\geq\frac{\log p_\beta}{\log p_\alpha}(\ell_\beta(X\upto n)-c).
\]
($\leftarrow$): Suppose $X\not\in\MLR_{\mu_{p_\beta}}$.  Let $\{\U_n\}_{n\in\omega}$ be a $\mu_{p_\beta}$-Martin-L{\"o}f test that $X$ fails, and let $\{R_n\}_{n\in\omega}$ be a uniformly c.e.\ sequence of prefix-free sets such that $\U_n=\llbracket R_n\rrbracket$.  We define a function $F$ on $\bigcup_{n\geq 1} R_{n(\lceil\beta\rceil+1)+1}$:  for $\sigma\in R_{n(\lceil\beta\rceil+1)+1}$ (where this is the greatest such $n$), we set $F(\sigma)=\frac{\log p_\beta}{\log p_\alpha}(\ell_\beta(\sigma)-n)$, which is a computable real since $\alpha$ and $\beta$ are computable.

We verify that $F$ is an $\alpha$-i.c.m.  Again using the fact that $p_\alpha^{\log p_\beta/\log p_\alpha}=p_\beta$, we have

\begin{align*}
\sum_{n\geq 1}\sum_{\sigma\in R_{n(\lceil\beta\rceil+1)+1}} p_\alpha^{F(\sigma)} &= \sum_{n\geq 1}\sum_{\sigma\in R_{n(\lceil\beta\rceil+1)+1}} p_\alpha^{\frac{\log p_\beta}{\log p_\alpha}(\ell_\beta(\sigma)-n)}\\
&=\sum_{n\geq 1}\sum_{\sigma\in R_{n(\lceil\beta\rceil+1)+1}} p_\beta^{\ell_\beta(\sigma)-n}\\
&= \sum_{n\geq 1} p_\beta^{-n} \sum_{\sigma\in R_{n(\lceil\beta\rceil+1)+1}} p_\beta^{\ell_\beta(\sigma)}\\
&= \sum_{n\geq 1} p_\beta^{-n} \mu_{p_\beta}\!\left(\llbracket R_{n(\lceil\beta\rceil+1)+1}\rrbracket\right)\\
&\leq \sum_{n\geq 1} p_\beta^{-n} \left(\frac{1}{2}\right)^{n(\lceil\beta\rceil+1)+1}\\
&\leq \sum_{n\geq 1} p_\beta^{-n} \left(\frac{1}{2}\right)^{n(\beta+1)+1}\\
&\leq \sum_{n\geq 1} p_\beta^{-n} p_\beta^{n(\beta+1)+1}&&\text{(since $p_\beta\geq 1/2$)}\\
&= \sum_{n\geq 1} p_\beta^{n\beta+1}=1.&&\text{(by Lemma \ref{lem-kj1})}\\
\end{align*}

Next, the set $\{(\sigma,m): F(\sigma)< m\}$ is clearly c.e., as the function $F$ is computable.  By the minimality of $K^{(\alpha)}$ among all $\alpha$-i.c.m.s, there is some $d\in\omega$ such that for all $\sigma$, if $\sigma\in R_{n(\lceil\beta\rceil+1)+1}$ for some $n\geq 1$, then $K^{(\alpha)}(\sigma)\leq F(\sigma)+d$.  Since $X\in\bigcap_{n\in\omega} \U_n$, it follows that for every $n\geq 1$, $X\in \U_{n(\lceil\beta\rceil+1)+1}$.  For each $n\geq 1$, let $\hat n$ be the least integer such that
\[
\frac{\log p_\beta}{\log p_\alpha}n + d\leq\frac{\log p_\beta}{\log p_\alpha}\hat n.
\]
Then since $X\in \U_{\hat n(\lceil\beta\rceil+1)+1}$, there is some $m$ such that $X\upto m\in R_{\hat n(\lceil\beta\rceil+1)+1}$, so that
\begin{align*}
K^{(\alpha)}(X\upto m)&\leq F(X\upto m)+d\\
&\leq \frac{\log p_\beta}{\log p_\alpha}(\ell_\beta(X\upto m)-\hat n)+d\\
&\leq \frac{\log p_\beta}{\log p_\alpha}(\ell_\beta(X\upto m)-n),\\
\end{align*}
which yields the conclusion.
\end{proof}

One consequence of Theorem \ref{thm-gen-kj} is a characterization of $(\alpha,\beta)$-incompressibility when $\alpha=\beta$.

\begin{cor}\label{cor-kincomp}
For any $X\in 2^\omega$ and computable $\alpha\geq 1$, $X\in\MLR_{\mu_{p_\alpha}}$ if and only if for all $n$,
\[
K^{(\alpha)}(X\upto n)\geq \ell_\alpha(X\upto n)-O(1).
\]
\end{cor}

We now conclude this section by showing that there are sequences that are $(\alpha,1)$-incompressible and yet not random with respect to any computable measure.  To do so, we use the following generalization of the above-discussed theorem of Davie's, Theorem \ref{thm-Davie}, which is an immediate consequence of \cite[Theorem 5.2.3]{HoyRoj09} due to Hoyrup and Rojas.  For $p\in(0,1)$, let $\mathcal{K}_c^p=\{X\in\cs\colon(\forall n)K(X\upto n)\geq -\log\mu_p(X\upto n)-c\}$.

\begin{thm}\label{thm-hr}
For any $c\in\omega$ and $\epsilon>0$, we can effectively find $n(c, \epsilon)\in\omega$ such that if $X\in \mathcal{K}^p_c$ then for all $n>n(c,\epsilon)$,
\[
\Biggl|\frac{\#_0(X\upto n)}{n}-p\Biggr|<\epsilon.
\]
\end{thm}

The proof of Theorem \ref{thm-hr} also uses the fact that the sequence $(\mathcal{K}_c^p)_{c\in\omega}$ defines a universal $\mu_p$-Martin-L\"of test.  That this fact can be used in tandem with \cite[Theorem 5.2.3]{HoyRoj09} is noted in Section 5.2.4 of \cite{HoyRoj09}.  As there are no new ideas involved in the proof, we leave the details to the reader.

Just as Theorem \ref{thm-Davie} applies to sufficiently long incompressible strings, so too does Theorem \ref{thm-hr} apply to sufficiently long $k$-incompressible strings.  We first need an auxiliary lemma.

\begin{lem}\label{lem-lk-incomp}
For computable $\alpha\geq 1$, there is some computable function $f$ such that for every $e\in\omega$ and every $\sigma\in\str$, if
$K^{(\alpha)}(\sigma)\geq \ell_\alpha(\sigma)-e$ then $K(\sigma)\geq-\log\mu_{p_\alpha}(\sigma)-f(e)$.
\end{lem}

\begin{proof}
Note that the collection $\{R_n\}_{n\in\omega}$  defined by
\[
R_n=\left\{\sigma\in\str:K(\sigma)<-\log\mu_{p_\alpha}(\sigma)-n\right\}
\]
defines a universal $\mu_{p_\alpha}$-Martin-L\"of test by the Levin-Schnorr theorem.   Applying the $(\leftarrow)$ direction of the proof of Theorem \ref{thm-gen-kj} to $\{R_n\}_{n\in\omega}$ yields a $d\in\omega$ such that for every $\sigma\in R_{n(\lceil\alpha\rceil+1)+1}$, $K^{(\alpha)}(\sigma)\leq\ell_\alpha(\sigma)-n+d$. Taking the converse, we have for each $e\in\omega$ and every $\sigma$, $K^{(\alpha)}(\sigma)\geq \ell_\alpha(\sigma)-e$ implies $K(\sigma)\geq -\log\mu_{p_\alpha}(\sigma)-((e+d)(\lceil\alpha\rceil+1)+1)$.

Setting $f(e)=(e+d)(\lceil\alpha\rceil+1)+1$ yields the desired function.
\end{proof}

We now state and prove the finitary version of Theorem \ref{thm-hr}. To do so, we need to generalize the notion of a ``shortest description" to $\alpha$-complexity.  For each $\sigma$, let $\sigma^{*(\alpha)}$ be the length-lexicographically least string such that $U(\sigma^{*(\alpha)})=\sigma$ and $K^{(\alpha)}(\sigma)=\ell_\alpha(\sigma^{*(\alpha)})$.

\begin{lem}\label{lem-davie2}
For computable $\alpha\geq 1$ and $\epsilon\in(0,1)$, there is some $d$ and $n(d,\epsilon)$ such that for all strings $\sigma$ such that $|\sigma^{*(\alpha)}|\geq n(d,\epsilon)$, 
\[
\Biggl|\frac{\#_0(\sigma^{*(\alpha)})}{|\sigma^{*(\alpha)}|}-p_\alpha\Biggr|<\epsilon.
\]
\end{lem}

\begin{proof}  Fix $\alpha$ and $\epsilon$ as above.  As in the proof of Lemma \ref{lem-davie}, let $c$ be the coding constant of $U\circ U$.  One can easily verify that for each $\sigma\in\str$, $K^{(\alpha)}(\sigma^{*(\alpha)})\geq \ell_\alpha(\sigma^{*(\alpha)})-c$.  By Lemma \ref{lem-lk-incomp}, it follows that $K(\sigma^{*(\alpha)})\geq -\log\mu_{p_\alpha}(\sigma^{*(\alpha)})-f(c)$.  Setting $d=f(c)$ and applying Theorem \ref{thm-hr} to any string $\sigma$ such that $|\sigma^{*(\alpha)}|\geq n(d,\epsilon)$ yields the desired conclusion.
\end{proof}

\begin{thm}\label{thm-k1}
For computable $\alpha\geq 1$ and $\epsilon\in(0,1-p_\alpha)$, if for all $n$,
\[
K(X\upto n)\geq \frac{n}{1+(\alpha-1)(1-p_\alpha-\epsilon)}-O(1),
\]
then $X$ is $(\alpha,1)$-incompressible.
\end{thm}

\begin{proof}
Given $X$ as in the hypothesis, let $d$ and $n(d,\epsilon)$ be as in Lemma \ref{lem-davie2}.  Then for all $n$ such that $K^{(\alpha)}(X\upto n)\geq n(d,\epsilon)$ and for all $\tau$ such that $U(\tau)=X\upto n$, we have 
\begin{equation}\label{eq-ugly-ineq}
|\tau|\geq  \frac{n}{1+(\alpha-1)(1-p_\alpha-\epsilon)}-j
\end{equation}
for some $j\in\omega$. In particular, if $\tau=(X\upto n)^{*(\alpha)}$, then by Lemma \ref{lem-davie2}, we have
\[
\Biggl|\frac{\#_0(\tau)}{|\tau|}-p_\alpha\Biggr|<\epsilon.
\]
This implies that 
\[
\Biggl|\frac{\#_1(\tau)}{|\tau|}-(1-p_\alpha)\Biggr|<\epsilon,
\]
and hence that
\begin{equation}\label{eq-kj3}
\#_1(\tau)\geq (1-p_\alpha-\epsilon)|\tau|.
\end{equation}
By the choice of $\tau$, we have
\[
K^{(\alpha)}(X\upto n)=|\tau|+(\alpha-1)\#_1(\tau),
\]
which when combined with (\ref{eq-kj3}) yields
\begin{equation}\label{eq-ugh}
K^{(\alpha)}(X\upto n)\geq (1+(\alpha-1)(1-p_\alpha-\epsilon))|\tau|.
\end{equation}
Combining (\ref{eq-ugly-ineq}) and (\ref{eq-ugh}) gives
\[
K^{(\alpha)}(X\upto n)\geq n-O(1).
\]

\end{proof}

\begin{cor}
For computable $\alpha>1$, there is an $(\alpha,1)$-incompressible sequence that is not Martin-L\"of random with respect to any computable measure.
\end{cor}

\begin{proof}
Choose $\epsilon$ such that $p_\alpha+\epsilon<1$ and let $\delta>0$ satisfy
$\delta\leq (\alpha-1)(1-p_\alpha-\epsilon)$.  By Miller \cite{Mil11}, for every $q\in (0,1)$, there is some $X\in\cs$ such that (i) $K(X\upto n)\geq q n-O(1)$ and (ii) $X$ does not compute any sequence $Y$ satisfying $K(Y\upto n)\geq q'n-O(1)$ for $q'>q$ in $(0,1]$.  Let $q=\frac{1}{1+\delta}$, and let $X$ satisfy (i) and (ii) for this choice of $q$.  First, we have
\[
K(X\upto n)\geq q n-O(1)=\frac{n}{1+\delta}-O(1)\geq \frac{n}{1+(\alpha-1)(1-p_\alpha-\epsilon)}, 
\]
and thus by Theorem \ref{thm-k1}, $X$ is $(\alpha,1)$-incompressible.  Next, it follows from (ii) that $X$ does not compute a Martin-L\"of random sequence.  Since every sequence that is random with respect some computable measure must compute a Martin-L\"of random sequence (by a result due independently to Zvonkin/Levin \cite{ZvoLev70} and Kautz \cite{Kau91}), it follows that $X$ is not random with respect to any computable measure.

\end{proof}

Our analysis of $(\alpha,\beta)$-incompressibility for computable $\alpha$ and $\beta$ is not complete, as the following is still open.

\begin{que}
For computable $\alpha\geq 1$ and $\beta>1$ with $\alpha> \beta$, is there an $(\alpha,\beta)$-incompressible sequence that is not Martin-L\"of random with respect to any computable measure?
\end{que}


%
%

\section{Effective dimension and generalized length functions}\label{sec-dim}

In this final section we consider effective Hausdorff dimension, effective packing dimension, and entropy in the context of generalized length functions. As stated at the end of Section \ref{sec-background}, it follows from the effective version of the Shannon-McMillan-Breiman theorem established by Hoyrup \cite{Hoy12} that for 
each computable Bernoulli measure $\mu_p$  and each $X\in \MLR_{\mu_p}$,
\[
\lim_{n\rightarrow\infty}\frac{K(X\upto n)}{n}=-p\log(p)-(1-p)\log(1-p)=h(p).
\]
We obtain the following generalization:

%
%

\begin{thm}\label{thm-j-k-entropy}
For every computable $\alpha,\beta\geq 1$, and all $X\in\MLR_{\mu_{p_\beta}}$,
\[
\lim_{n\to\infty} \frac{K^{(\alpha)}(X\upto n)}{n}=-\frac{1}{\log p_\alpha} h(p_\beta).
\]
\end{thm}

\begin{proof}
Fix computable $\alpha,\beta\geq 1$.  First, we will use the well-known fact that law of large numbers holds for all sequences that are random with respect to a Bernoulli measure; that is, given a computable real $p\in(0,1)$ and the Bernoulli $p$-measure $\mu_p$, if $X\in\MLR_{\mu_p}$ then
\[
\lim_{n\to\infty} \frac{\#_0(X\upto n)}{n}=p\text{  and  }\lim_{n\to\infty} \frac{\#_1(X\upto n)}{n}=1-p.
\]

\noindent With this in mind, let $X\in\MLR_{\mu_{p_\beta}}$.  We first show that the effective $\alpha$-packing and effective $\alpha$-Hausdorff dimensions of $X$ are equal.  We handle effective $\alpha$-packing dimension first.  By Theorem \ref{thm-gen-kj}, since $X\in\MLR_{\mu_{p_\beta}}$, there is some $c\in\omega$ such that for every $n\in\omega$,

\[
K^{(\alpha)}(X\upto n)\geq \frac{\log p_\beta}{\log p_\alpha}(\ell_\beta(X\upto n)-c).
\]

\noindent Hence\[
\liminf_{n\to\infty} \frac{K^{(\alpha)}(X\upto n)}{n}\geq\liminf_{n\to\infty}\frac{\log p_\beta}{\log p_\alpha}\left(\frac{\ell_\beta(X\upto n)-c}{n}\right)=\lim_{n\to\infty}\frac{\log p_\beta}{\log p_\alpha}\left(\frac{\ell_\beta(X\upto n)-c}{n}\right),
\]
assuming the latter limit exists.  We verify that it does as follows:
\begin{align*}
\lim_{n\to\infty}\frac{\log p_\beta}{\log p_\alpha}\left(\frac{\ell_\beta(X\upto n)-c}{n}\right) &=\frac{\log p_\beta}{\log p_\alpha}\cdot\lim_{n\to\infty}\frac{\#_0(X\upto n)+\beta\cdot\#_1(X\upto n)-c}{n}\\
&=\frac{\log p_\beta}{\log p_\alpha}\left(p_\beta+\beta(1-p_\beta)\right),
\end{align*}
where the latter equality follows from the law of large numbers.
Next, by Theorem \ref{thm-K-upper-bound}, let $c'\in\omega$ be such that for every $n\in\omega$,
\begin{equation}\label{eq-dim1}
K^{(\alpha)}(X\upto n)\leq \frac{\log p_\beta}{\log p_\alpha}\left(\ell_\beta(X\upto n)+K^{(\beta)}(\lfloor\ell_\beta(X\upto n)\rfloor)\right)+c'.
\end{equation}
Lemma \ref{lem-K-inequality} gives us that there is some $c''\in\omega$ such that for every $n\in\omega$,
\begin{equation}\label{eq-dim2}
K^{(\beta)}(\lfloor\ell_\beta(X\upto n)\rfloor)\leq (\beta+1)\log (\lfloor\ell_\beta(X\upto n)\rfloor)+c''.
\end{equation}
Lastly, for all $n\in\omega$, 
\begin{equation}\label{eq-dim3}
\lfloor\ell_\beta(X\upto n)\rfloor\leq\ell_\beta(X\upto n)\leq \beta n,
\end{equation}
since we can have at most $n$ 0's in any string of length $n$.
Thus we have
\begin{align*}
\limsup_{n\to\infty} \frac{K^{(\alpha)}(X\upto n)}{n} &\leq \limsup_{n\to\infty}\left[ \frac{\log p_\beta}{\log p_\alpha}\left(\frac{\ell_\beta(X\upto n)+K^{(\beta)}(\lfloor\ell_\beta(X\upto n)\rfloor)}{n}\right)+\frac{c'}{n}\right]&&\text{by (\ref{eq-dim1})}\\
&\leq \limsup_{n\to\infty} \left[\frac{\log p_\beta}{\log p_\alpha}\left(\frac{\ell_\beta(X\upto n)+(\beta+1)\log(\lfloor\ell_\beta(X\upto n)\rfloor)+c''}{n}\right)+\frac{c'}{n}\right]&&\text{by (\ref{eq-dim2})}\\
&\leq \limsup_{n\to\infty} \left[\frac{\log p_\beta}{\log p_\alpha}\left(\frac{\ell_\beta(X\upto n)+(\beta+1)\log \beta n+c''}{n}\right)+\frac{c'}{n}\right]&&\text{by (\ref{eq-dim3})}\\
&=\lim_{n\to\infty} \left[\frac{\log p_\beta}{\log p_\alpha}\left(\frac{\ell_\beta(X\upto n)+(\beta+1)\log \beta n+c''}{n}\right)+\frac{c'}{n}\right],\\
\end{align*}
again assuming that the latter limit exists.  Now, following similar steps as those taken with the first limit, we see that
\[
\begin{split}
&\lim_{n\to\infty} \left[\frac{\log p_\beta}{\log p_\alpha} \left(\frac{\ell_\beta(X\upto n)+(\beta+1)\log \beta n+c''}{n}\right)+\frac{c'}{n}\right]\\
 &=\frac{\log p_\beta}{\log p_\alpha}\left[\lim_{n\to\infty}\frac{\ell_\beta(X\upto n)}{n}+\lim_{n\to\infty}\frac{(\beta+1)\log \beta n}{n}+\lim_{n\to\infty}\frac{c''}{n}\right]+\lim_{n\to\infty}\frac{c'}{n}\\
&=\frac{\log p_\beta}{\log p_\alpha}(p_\beta+\beta(1-p_\beta)).\\
\end{split}
\]
 It follows that
\[
\limsup_{n\to\infty} \frac{K^{(\alpha)}(X\upto n)}{n}=\liminf_{n\to\infty} \frac{K^{(\alpha)}(X\upto n)}{n}=\lim_{n\to\infty} \frac{K^{(\alpha)}(X\upto n)}{n}=\frac{\log p_\beta}{\log p_\alpha}(p_\beta+\beta(1-p_\beta)),
\]
as desired.  It now remains to show that
\[
\lim_{n\to\infty} \frac{K^{(\alpha)}(X\upto n)}{n}=-\frac{1}{\log p_\alpha}h(p_\beta).
\]
This is straightforward:
\[
\begin{split}
h(p_\beta) &= -p_\beta\log p_\beta-(1-p_\beta)\log(1-p_\beta)= -p_\beta\log p_\beta-(1-p_\beta)\log p^\beta_\beta\\
&= -p_\beta\log p_\beta-\beta(1-p_\beta)\log p_\beta= -(\log p_\beta)(p_\beta+\beta(1-p_\beta)).
\end{split}
\]
The full result immediately follows.
\end{proof}

We conclude by noting that $\lim_{\alpha\rightarrow\infty}-1/\log(p_\alpha) =\lim_{p_\alpha\rightarrow 1}-1/\log(p_\alpha)=\infty$, and hence the $\alpha$-dimension of a sequence does not have an upper bound of 1, unlike the case for $\alpha=1$.

%
%
%

\bibliographystyle{alpha}
\bibliography{genKC}

\end{document}